\documentclass[conference,a4paper,10pt]{IEEEtran}

\usepackage[utf8]{inputenc} 
\usepackage[T1]{fontenc}
\usepackage{url}
\usepackage{ifthen}
\usepackage{cite}
\usepackage[cmex10]{amsmath}
\usepackage{graphicx,subfigure}
\usepackage{amssymb}
\usepackage{dsfont}
\usepackage{enumitem}
\usepackage{color}
\usepackage{epstopdf}
\usepackage{flushend}
\usepackage{lipsum}
\usepackage[font=small]{caption}

\newtheorem{theorem}{Theorem}
\newtheorem{lemma}{Lemma}

\def \F {{\cal F}}

\def \X {{\cal X}}

\newcommand{\dff}{\stackrel{\scriptscriptstyle\triangle}{=}}
\newcommand{\bbe}{\mathbb{E}}
\def \P{\mathbb{P}}

\def \S{{\cal S}}

\def \G{{\cal G}}
\def \L{{\cal L}}

\begin{document}

\title{Quickest Search for a Change Point}

\author{
	\IEEEauthorblockN{Javad Heydari and Ali Tajer}
	\IEEEauthorblockA{Electrical, Computer, and Systems Engineering Department\\
	Rensselaer Polytechnic Institute
%	Troy, NY, USA
	}
}

\maketitle

\begin{abstract}
This paper considers a sequence of  random variables generated according to a {\em common} distribution. The distribution might undergo periods of transient changes at an unknown set of time instants, referred to as change-points. The objective is to sequentially collect measurements from the sequence and design a dynamic decision rule for the quickest {\em identification} of one change-point in real time, while, in parallel, the rate of false alarms is controlled. This setting is different from the conventional change-point detection settings in which there exists at most one change-point that can be either persistent or transient. The problem is considered under the minimax setting with a constraint on the false alarm rate before the first change occurs. It is proved that the Shewhart test achieves {\em exact} optimality under worst-case change points and also worst-case data realization. Numerical evaluations are also provided to assess the performance of the decision rule characterized.
\end{abstract}

%\textit{A full version of this paper is accessible at:}
%\url{http://www.isit2018.org/} 

\section{Introduction}
%\vspace{-.1 in}

%{\renewcommand{\thefootnote}{}\footnotetext{\noindent This research was supported in part by the U. S. National Science Foundation under Grants ECCS-1455228 and the CAREER Award ECCS-1554482.}}

Real-time monitoring of a system or process for identifying a change of behavior arises in many application domains such as detecting faults or security breaches in networks, and searching for under-utilized spectrum bands for opportunistic spectrum access. It is often of interest to detect abrupt changes with minimal delay after they occur. At the same time designing detection rules that are too sensitive to changes in observations are susceptible to raising frequent false alarms. This creates an inherent tension between the quickness and the quality of the decisions. 

In classical quickest change-point detection, the process under consideration is a sequence of random variables, distribution of which changes at an unknown time instant permanently~\cite{PoorQuick}. A decision maker aims to design a stopping rule to detect such a change with the minimal average delay by monitoring it sequentially, while, in parallel, controlling the rate of false alarms. The setting and objective of this paper has major distinctions from the classical quickest change-point detection. First the change is not persistent, i.e., the distribution of the sequence returns to the pre-change distribution after the change. Secondly, multiple changes occur throughout the monitoring process. Furthermore, the goal of this paper is to search for one of the change-points and detect it immediately after it occurs, while in the classical setting the objective is to minimize the average decision delay. The drawback of minimizing the average decision delay is that it allows for arbitrary large delay~\cite{PoorQuick}. Therefore, in this paper, similar to~\cite{Bojdecki,Pollak-Krieger,Moustakides:2014,Moustakides:2015}, a probability maximizing approach is adopted. In this approach the objective is to design a stopping rule that maximizes the probability of stopping at a change-point.

Quickest detection of transient changes in a sequence has gained research interest in recent years. In~\cite{Ebi,Nikiforov}, the problem of detecting one transient change is considered. The study in~\cite{Ebi} aims to characterize the shortest duration of a change that can be detected as the false alarm rate goes to zero, while~\cite{Nikiforov} treats a detection when the transient change is over as a missed detection and aims to minimize it. The studies in~\cite{GM-VV,GF-VV,VKM,VKM2} consider a setting in which the change does not occur abruptly, but rather through a series of changes, after which it settles to a permanent steady state. In this setting, the steady-state distribution is different from the pre-change one, while the pre-change and post-change models are identical in~\cite{Ebi,Nikiforov}. In~\cite{GM-VV} the transient duration is a single measurement, while in~\cite{GF-VV} it is a deterministic unknown constant. The data model of this paper is similar to that of~\cite{Ebi,Nikiforov} in the sense that the pre-change and post-change distributions are the same. However, in this paper the sequence experiences {\em multiple} transient changes. Quickest change-point detection under multiple transient changes is also considered in~\cite{VKM,VKM2}, where the state of the system is assumed to be a Markov process and only one of the states, which is also an absorbing state, is considered as the desirable change state. 
%Hence, there exists a distribution associated to each state, and the process can be in any of these states prior to reaching the absorbing (change) state. 
The oscillatory behavior of the sequence under consideration in this paper between two distributions (pre-change and post-change) is its fundamental difference with the aforementioned studies.

Besides the distinction in the data model, the ultimate goal of this paper differs from the classical settings. Instead of minimizing the average detection delay, a probability maximization approach is adopted in order to maximize the probability of stopping at a change-point. This approach was first proposed in~\cite{Bojdecki} in a Bayesian setting for detecting a persistent change in a sequence of independent and identically distributed (i.i.d.) random variables. The results were extended to dependent random variables~\cite{Sarnowski}, and composite post-change model~\cite{Pollak-Krieger}. In~\cite{Moustakides:2014,Moustakides:2015}, the objective is detecting a persistent change immediately by using the first measurement under the change state. Under both Bayesian and minimax regimes the exactly optimal detection rules have been characterized, and the results have been extended to independent non-identically distributed measurements, composite post-change models~\cite{Moustakides:2014}. The extension to Markovian measurements is studied in~\cite{Moustakides:2015}.

The remainder of the paper is organized as follows. Section~\ref{sec:model} provides the data model  and formalizes the search problem of interest. The quickest search rule is characterized in Section~\ref{sec:sol}, where the performance bounds and the associated stopping rule that achieves this bound are specified. Section~\ref{sec:sim} provides the numerical evaluation of the quickest search approach, and concluding remarks are provided in Section~\ref{sec:conclusion}.

%\vspace{-.1 in}
\section{Model and Formulation}
\label{sec:model}

\subsection{Data Model} \label{sec:data_model}
%\vspace{-.05 in}

Consider a sequence of random variables denoted by $\mathcal{X}\dff\{X_t:t\in\mathbb{N}\}$.
As shown in Fig.~\ref{fig:model}, these random variables have a common probability distribution that undergoes periods of {\em transient} changes at an unknown and non-random set of time instants. Specifically, the elements of $\X$ are {\em nominally} generated according to a probabilistic distribution with the cumulative density function (cdf) $F_0$. However, there potentially exist a finite but unknown number of time instants $\gamma\dff\{\gamma_i:i\in\{1,\dots,s\}\}$, referred to as change-points, at which the distribution changes from the nominal cdf $F_0$ to a distinct one with cdf $F_1$. It is assumed that $s\in\mathbb{N}$ is unknown, and the duration of each transient change is a known constant denoted by $T$, and the transient intervals are assumed to be non-overlapping, i.e., $|\gamma_i-\gamma_j|>T$ for all $i,j\in\{1,\dots,s\}$. We define $\S$ as the set of all instants $t\in\mathbb{N}$ at which $X_t$ is generated by $F_1$, i.e., 
\begin{align}
\S\dff \{t\;:\; X_t\sim F_1\}\ .
\end{align}
Hence, for the elements of $\X$ we have the dichotomous model
\begin{equation}\label{eq:H}
    \begin{array}{cl}
      X_t\sim F_0\;, & t\in\mathbb{N}\backslash \S \\
      X_t\sim F_1\;, & t\in\S\ 
    \end{array}\ .
\end{equation}
We also assume that there exist well-defined probability density functions (pdfs) corresponding to $F_0$ and $F_1$, which we denote by $f_0$ and $f_1$, respectively. Subsequently, we denote the probability measure governing sequence $\X$ and the expectation with respect to this measure by $\P_\gamma$ and $\bbe_\gamma$, respectively. We also use $\P_\infty$ and $\bbe_\infty$ for the case that no change occurs in the data under consideration, i.e., $s=0$, and the distribution is always $F_0$.

%\vspace{-.1 in}
\subsection{Problem Formulation}
%\vspace{-.05 in}

The objective is to sequentially collect measurements from the sequence $\X$ and design a sequential decision rule for the quickest {\em identification} of one change-point, i.e., one of the elements in $\gamma=\{\gamma_i:i\in\{1,\dots,s\}\}$ in real time, while, in parallel, the rate of false alarms is controlled. Hence, the sequential decision-making process continually collects measurements until the stopping time of the process, at which point it is confident enough that the last collected sample belongs to the set $\gamma$. It is noteworthy that the setting in which there exists only one change-point, which can be either persistent or transient, is studied extensively in the literature (c.f.~\cite{PoorQuick}--\cite{GM-VV}). In contrast, in this paper we assume that the number of change-points and the ensuing transient intervals can exceed one.  

The information generated by the data sequentially up to time $t$ generates the filtration $\{\F_t\,:\,t\in\mathbb{N}\}$ where
\begin{align}
\F_t\dff\sigma\big(X_1,X_2,\dots,X_t\big)\ .
\end{align}
Furthermore, we also define a coarser filtration, which at time $t\in\mathbb{N}$ is generated by only the measurements from the last change-point up to time $t$. This filtration is denoted by
\begin{align}
\G_t\dff\sigma\big(X_{r(t)+1},X_{r(t)+2},\dots,X_t\big)\ ,
\end{align}
where we have defined $r(t)\dff\sup\ \{i\in\S:i\leq t\}$, and adopt the convention that the supremum of an empty set is zero.
The sequential sampling process continues until the stopping time, denoted by $\tau$, after which no further measurements are made and a change is declared. The stopping time $\tau$ is set to be a $\G_t$-measurable function.

\begin{figure}[t!]
\centering
\includegraphics[width=0.45\textwidth]{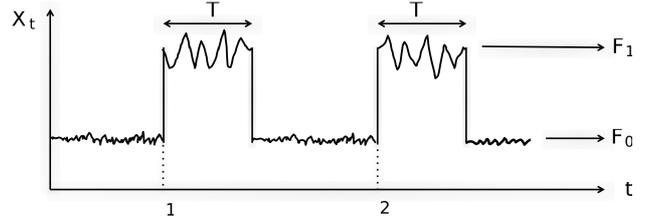}
%\vspace{-0.15in}
\caption{Data model.}
\label{fig:model}
%\vspace{-0.17in}
\end{figure}

 Two relevant performance measures for evaluating the quality of these sampling and decision-making processes are the {\em agility} of the process as well as the frequency of false alarms. To account for the agility, since we are interested in the real-time identification of a change-point, the conventional average detection delay is ineffective as it does not impose a hard threshold on the detection delay. To circumvent this, for quantifying the agility of the process we adopt a probability-based approach
as also done by~\cite{Bojdecki}~and~\cite{Moustakides:2014}. Specifically, we investigate two minimax settings in which we formalize probability maximization criteria mimicking Pollak's~\cite{Pollak} and Lorden's~\cite{Lorden} approach. In particular, we define a Pollak-like criterion as
\begin{align}\label{eq:Pollak_criterion}
\mathcal{L}_{\rm P}(\tau)\dff \inf_{\gamma}\ \sum_{\gamma_i\in \gamma}\ \mathbb{P}_\gamma(\tau=\gamma_i\;|\;\tau\geq \gamma_i)\ ,
\end{align}
where the infimum is over all possible realizations of the unknown set $\gamma$. Similarly, we define a Lorden-like worst case criterion as
\begin{align}\label{eq:Lorden_criterion}
\mathcal{L}_{\rm L}(\tau)\dff \inf_{\gamma}\ \sum_{\gamma_i\in \gamma}\ \underset{\mathcal{F}_{\gamma_i-1}}{\text{essinf}}\ \mathbb{P}_\gamma(\tau=\gamma_i\;|\;\mathcal{F}_{\gamma_i-1},\tau\geq \gamma_i)\ .
\end{align}
It can be readily verified that
\begin{align}
\mathcal{L}_{\rm L}(\tau)\leq\mathcal{L}_{\rm P}(\tau)\ .
\end{align}
%It is noteworthy that Pollak's criterion is more natural for modeling an unknown change-point with no random mechanism specifying it (in contrast to Bayesian settings), while Lorden's criterion is more apt when the change-point mechanism depends on the history of the observations. 
In order to account for the frequency of the false alarms, we use $\bbe_\infty\{\tau\}$, which captures the average run length to a false alarm before the first change-point $\gamma_1$ occurs.

There exists an inherent tension between the rate of false alarms on the one hand, and the measures $\mathcal{L}_{\rm P}(\tau)$ and $\mathcal{L}_{\rm L}(\tau)$, on the other hand as improving these two measures penalizes the false alarm rate. Leveraging such tension, an optimal sampling strategy can be obtained by balancing false alarm rate and the detection probability. Hence, under the Pollak-like criterion in~\eqref{eq:Pollak_criterion} the sampling strategy is the solution to
\begin{align}\label{eq:p1}
\begin{array}{ll}
\sup_{\tau} & \mathcal{L}_{\rm P}(\tau)\\
\text{s.t.} & \bbe_\infty\{\tau\}\geq\eta
\end{array}\ ,
\end{align}
and under the Lorden-like criterion in~\eqref{eq:Lorden_criterion} it is the solution to
\begin{align}\label{eq:p2}
\begin{array}{ll}
\sup_{\tau} & \mathcal{L}_{\rm L}(\tau)\\
\text{s.t.} & \bbe_\infty\{\tau\}\geq\eta
\end{array}\ ,
\end{align}
where $\eta\geq 1$ in both settings controls the false alarm rate.

%\vspace{-.1 in}
\section{Quickest Search Rules}
\label{sec:sol}

 In this paper we characterize the optimal stopping rules for the problems in~\eqref{eq:p1} and \eqref{eq:p2}. For this purpose, we first find upper bounds on the objective functions $\mathcal{L}_{\rm P}(\tau)$ and $\mathcal{L}_{\rm L}(\tau)$ in Section~\ref{sec:bound}. Then we briefly review the Shewhart test in Section~\ref{sec:Shewhart}, and in Section~\ref{sec:optimal} we show that by using the Shewhart test as the decision rule the values of  $\mathcal{L}_{\rm P}(\tau)$ and $\mathcal{L}_{\rm L}(\tau)$ achieve their upper bounds established in Section~\ref{sec:bound}, thereby establishing that the Shewhart test is an optimal solution to~\eqref{eq:p1} and \eqref{eq:p2}

\subsection{Upper Bounds on the Objective Functions}
\label{sec:bound}

In order to facilitate finding upper bounds on the objective functions in~\eqref{eq:Pollak_criterion} and~\eqref{eq:Lorden_criterion}, we denote  the likelihood ratio of the measurement made at time $t$ by
\begin{align}\label{eq:LLR}
\ell_t\dff\frac{f_1(X_t)}{f_0(X_t)}\ .
\end{align}
The following theorem characterizes an upper bound for both  Lorden-like and modified Pollak-like criteria defined in \eqref{eq:Pollak_criterion}~and~\eqref{eq:Lorden_criterion}, respectively.
\begin{theorem}[Upper Bound]\label{th:bound}
For the objective functions $\mathcal{L}_{\rm P}(\tau)$ and $\mathcal{L}_{\rm L}(\tau)$ we have
\begin{align}\label{eq:up}
\mathcal{L}_{\rm L}(\tau) &\leq s\cdot\frac{\bbe_\infty\{\ell_\tau\}}{\bbe_\infty\{\tau\}} \ ,\\
\text{and}\quad\mathcal{L}_{\rm P}(\tau) &\leq s\cdot\frac{\bbe_\infty\{\ell_\tau\}}{\bbe_\infty\{\tau\}}\ .
\end{align}
\end{theorem}
\begin{proof}
We provide the proof for the  Loreden-like criterion $\mathcal{L}_{\rm L}(\tau)$. The proof for the Pollak-like criterion follows the same line of arguments, and is omitted for brevity. 

We start by considering that case that we only have one change-point, i.e., $s=1$. From the definition in~\eqref{eq:Lorden_criterion} we have
\begin{align}\label{eq:bound1}
\mathcal{L}_{\rm L}(\tau) \leq \mathbb{P}_\gamma(\tau=\gamma_1\;|\;\mathcal{F}_{\gamma_1-1},\tau\geq \gamma_1)\ .
\end{align}
By multiplying both sides of \eqref{eq:bound1} by $\mathds{1}_{\{\tau\geq\gamma_1\}}$, and taking the expectation with respect to the nominal measure $\P_\infty$ we obtain
\begin{align}
\label{eq:mul0}
&\bbe_\infty\{\mathds{1}_{\{\tau\geq\gamma_1\}}\mathcal{L}_{\rm L}(\tau)\} \\
\label{eq:mul1}
&\leq \bbe_\infty\{\mathds{1}_{\{\tau\geq\gamma_1\}}\mathbb{P}_\gamma(\tau=\gamma_1\;|\;\mathcal{F}_{\gamma_1-1},\tau\geq \gamma_1)\}\\
\label{eq:mul2} 
& = \bbe_\infty\{\mathbb{E}_\gamma\{\mathds{1}_{\{\tau=\gamma_1\}}\;|\;\mathcal{F}_{\gamma_1-1},\tau\geq \gamma_1\}\}\\
\label{eq:mul3} 
&= \bbe_\infty\Big\{\mathbb{E}_\infty\Big\{\mathds{1}_{\{\tau=\gamma_1\}}\prod_{t\in\S,t\leq\tau}\ell_t\;|\;\mathcal{F}_{\gamma_1-1},\tau\geq \gamma_1\Big\}\Big\}\\
%&= \bbe_\infty\Big\{\mathbb{E}_\infty\Big\{\ell_\tau\prod_{\gamma\in\S,\gamma<\tau}\ell_{\gamma}\;|\;\mathcal{F}_{\gamma_1-1},\tau\geq \gamma_1\Big\}\Big\}\\
\label{eq:mul4} 
&= \bbe_\infty\Big\{\mathds{1}_{\{\tau=\gamma_1\}}\prod_{t\in\S,t\leq\tau}\ell_t\Big\}\\
\label{eq:mul5} 
& = \mathbb{E}_\infty\{\mathds{1}_{\{\tau=\gamma_1\}}\ell_{\gamma_1}\}\ ,
\end{align}
where~\eqref{eq:mul1} is due to the definition of $\mathcal{L}_{\rm L}(\tau)$,~\eqref{eq:mul2} holds since $\mathds{1}_{\{\tau\geq\gamma_1\}}$ is measurable with respect to $\mathcal{F}_{\gamma_1-1}$ and  the event $\{\tau=\gamma_1\}$ is a subset of $\{\tau\geq\gamma_1\}$,~\eqref{eq:mul3} results from changing the expectation measure,~\eqref{eq:mul4} is due to the towering property of expectation, and~\eqref{eq:mul5} holds since $\tau$ is $\G_t$-measurable. On the other hand, since $\mathcal{L}_{\rm L}(\tau)$ is deterministic, the term in~\eqref{eq:mul0} can be expanded to
\begin{align}
\bbe_\infty\{\mathds{1}_{\{\tau\geq\gamma_1\}}\mathcal{L}_{\rm L}(\tau)\} 
&=\mathbb{P}_\infty\{\tau\geq\gamma_1\}\mathcal{L}_{\rm L}(\tau)\ ,
\end{align}
which in conjunction with~\eqref{eq:mul0} and~\eqref{eq:mul5} establishes that
\begin{align}\label{eq:mul6}
\mathbb{P}_\infty\{\tau\geq\gamma_1\}\mathcal{L}_{\rm L}(\tau) \leq \mathbb{E}_\infty\{\ell_{\gamma_1}\mathds{1}_{\{\tau=\gamma_1\}}\}\ .
\end{align}
Summing  both sides of \eqref{eq:mul6} over all $\gamma_1\in\mathbb{N}\cup\{0\}$ yields
\begin{align}
\mathcal{L}_{\rm L}(\tau)\mathbb{E}_\infty\{\tau\}
&\leq \mathbb{E}_\infty\{\ell_\tau\}\ ,
\end{align}
which concludes the desired result for the case of $s=1$. For any $s>1$ we have
\begin{align}
&\mathcal{L}_{\rm L}(\tau) \\
\label{l1}
&\leq \inf_{\gamma_s}\inf_{\gamma_1<\dots<\gamma_s} \sum_{i=1}^s\mathbb{P}_\gamma(\tau=\gamma_i\;|\;\mathcal{F}_{\gamma_i-1},\tau\geq \gamma_i)\\ 
\label{l2}
&\leq \frac{\mathbb{E}_\infty\left\{\ell_{\tau}\right\}}{\mathbb{E}_\infty\{\tau\}}
+\inf_{\gamma_1<\dots<\gamma_s} \sum_{i=1}^{s-1}\mathbb{P}_\gamma(\tau=\gamma_i\;|\;\mathcal{F}_{\gamma_i-1},\tau\geq \gamma_i)\\ 
\label{l3}
&= \frac{\mathbb{E}_\infty\left\{\ell_{\tau}\right\}}{\mathbb{E}_\infty\{\tau\}}
+\inf_{\gamma_1<\dots<\gamma_{s-1}} \sum_{i=1}^{s-1}\mathbb{P}_\gamma(\tau=\gamma_i\;|\;\mathcal{F}_{\gamma_i-1},\tau\geq \gamma_i)\\ 
\label{l4}
%&\leq \frac{\mathbb{E}_\infty\left\{\ell_{\tau}\right\}}{\mathbb{E}_\infty\{\tau\}}+
%\inf_{\gamma_2}\frac{\mathbb{E}_\infty\left\{\ell_{\tau}\;|\;\tau<\gamma_2\right\}}{\mathbb{E}_\infty\{\tau\;|\;\tau<\gamma_2\}}\\ 
&\leq s\cdot\frac{\mathbb{E}_\infty\left\{\ell_{\tau}\right\}}{\mathbb{E}_\infty\{\tau\}}\ ,
\end{align}
where~\eqref{l2} is due to the result we obtained from case $s=1$,~\eqref{l3} holds since the remaining terms are independent of $\gamma_s$, and~\eqref{l4} results from applying induction.
Generally, when we have multiple change-points, for any $i\in\{1,\dots,s\}$ if we define
\begin{align}
\mathcal{L}^i_{\rm L}(\tau)\dff \mathbb{P}_\gamma(\tau=\gamma_i\;|\;\mathcal{F}_{\gamma_i-1},\tau\geq \gamma_i)\ ,
\end{align}
by following the same line of argument as the case of $s=1$ we can show that for every $i\in\{1,\dots,s\}$ we have
\begin{align}
\mathcal{L}^i_{\rm L}(\tau) \leq\frac{\mathbb{E}_\infty\{\ell_\tau\}}{\mathbb{E}_\infty\{\tau\}}\ .
\end{align}
%Now, by noting that
%\begin{align}
%\mathcal{L}_{\rm L}(\tau) \leq \sum_{i=1}^s \mathcal{L}^i_{\rm L}(\tau)\ ,
%\end{align}
which concludes the proof.
\end{proof}

\subsection{Shewhart Test}
\label{sec:Shewhart}

%\begin{figure}[t!]
%\centering
%\includegraphics[width=0.35\textwidth]{Shewhart.pdf}
%%\vspace{-0.15in}
%\caption{The probability of detecting the first or any change-point.}
%\label{fig:Shewhart}
%%\vspace{-0.17in}
%\end{figure}

The form of Shewhart test that we adopt in this paper consists in a dynamic and sequential likelihood ratio test. Formally, at each time $t$ based on the observation $X_t$ we form the likelihood ratio value $\ell_t$ defined in~\eqref{eq:LLR}. The Shewhart test compares $\ell_t$ with a pre-specified and deterministic threshold $\alpha$ and declares a change when $\ell_t$ exceeds $\alpha$. Specifically, the stopping time of the Shewhart test is found via
\begin{align}\label{eq:stop}
\tau_{\rm s}\dff\inf\,\{t\;:\;\ell_t\geq\alpha\}\ .
\end{align}
The value of the threshold $\alpha$ is chosen such that the average run length to a false alarm is guaranteed to be not smaller than $\eta$, and can be computed from 
\begin{align}\label{eq:threshold}
\mathbb{P}_\infty(\ell_1\geq\alpha)=\eta^{-1}\ .
\end{align}

\subsection{Optimality of Shewhart Test}
\label{sec:optimal}

In this subsection we prove the exact optimality of the Shewhart test formalized in~\eqref{eq:stop} and~\eqref{eq:threshold} for both problems in~\eqref{eq:p1}~and~\eqref{eq:p2}. For this purpose, we start by proving that corresponding to any feasible\footnote{A decision rule with stopping time $\nu$ is called feasible if it satisfies the false alarm constrain, i.e., $\bbe_\infty\{\nu\}\geq\eta$.}  decision rule with the stopping time $\nu$ and the associated ratio
\begin{align}
\frac{\mathbb{E}_\infty\{\ell_\nu\}}{\mathbb{E}_\infty\{\nu\}}\ ,
\end{align}
we can construct an alternative feasible decision rule that  achieves the false alarm constraint with equality, and its stopping time, denoted by $\nu'$ achieves the same ratio, i.e., 
\begin{align}
\bbe\{\nu'\}=\eta \quad\mbox{and}\quad \frac{\mathbb{E}_\infty\{\ell_{\nu'}\}}{\mathbb{E}_\infty\{\nu'\}} =\frac{\mathbb{E}_\infty\{\ell_\nu\}}{\mathbb{E}_\infty\{\nu\}} \ .
\end{align}
This observation is formalized in the following lemma.
\begin{lemma}\label{lemma:eq}
Corresponding to any given feasible decision rule with the stopping time $\nu$ there always exists an alternative feasible decision rule that satisfies the false alarm constraint with equality, and its stopping time, denoted by $\nu'$, yields
\begin{align}\label{eq:ratio}
\frac{\mathbb{E}_\infty\{\ell_{\nu'}\}}{\mathbb{E}_\infty\{\nu'\}}=\frac{\mathbb{E}_\infty\{\ell_\nu\}}{\mathbb{E}_\infty\{\nu\}} \ .
\end{align}
\end{lemma}
\begin{proof}
Define $\pi_0$ as the probability of detecting a change without taking any measurement by the given stopping rule with the stopping time $\nu$. Then, it can be readily verified that
\begin{align}
\frac{\mathbb{E}_\infty\{\ell_\nu\}}{\mathbb{E}_\infty\{\nu\}} &= \frac{(1-\pi_0)\mathbb{E}_\infty\{\ell_\nu\;|\;\nu>0\}}{(1-\pi_0)\mathbb{E}_\infty\{\nu\;|\;\nu>0\}} \\
&= \frac{\mathbb{E}_\infty\{\ell_\nu\;|\;\nu>0\}}{\mathbb{E}_\infty\{\nu\;|\;\nu>0\}} \ .
\end{align}
Now, if corresponding to the stopping time $\nu$, the false alarm constraint does not hold with equality, i.e., if
\begin{align}
\bbe_\infty\{\nu\}=(1-\pi_0)\mathbb{E}_\infty\{\nu\;|\;\nu>0\}>\eta\ ,
\end{align}
then we design an alternative decision rule that (i) at every time $\nu>0$ it is similar to the given rule, which leads to 
\begin{align}
\label{eq:equal1}\mathbb{E}_\infty\{\nu'\;|\;\nu'>0\} & =\mathbb{E}_\infty\{\nu\;|\;\nu>0\}\ , \\
\label{eq:equal2}\mathbb{E}_\infty\{\ell_{\nu'}\;|\;\nu'>0\} & = \mathbb{E}_\infty\{\ell_\nu\;|\;\nu>0\}\ ,
\end{align}
and (ii) at $\nu=0$ the initial probability of detecting a change without collecting any measurements is set to $\pi'_0>\pi_0$, where $\pi'_0$ is the unique solution to 
\begin{align}\label{eq:equal3}
\bbe_\infty\{\nu'\}=(1-\pi_0')\mathbb{E}_\infty\{\nu\;|\;\nu>0\}=\eta\ .
\end{align}
Therefore, \eqref{eq:equal1}-\eqref{eq:equal3} collectively establish that $\nu'$ is feasible and achieves the same ratio specified in~\eqref{eq:ratio}.
\end{proof}
Next, we leverage the result of Lemma~\ref{lemma:eq} and prove the following  properties for the Shewhart test:
\begin{enumerate}
\item It is a feasible test.
\item Among all feasible tests, it maximizes the upper bound on $\mathcal{L}_{\rm P}(\tau)$ and $\mathcal{L}_{\rm L}(\tau)$ established in Theorem~\ref{th:bound}.
\item The objective functions $\mathcal{L}_{\rm P}(\tau)$ and $\mathcal{L}_{\rm L}(\tau)$ meet this maximum upper bound when using the Shewhart test.
\end{enumerate}
These properties are formalized in the following lemmas.
\begin{lemma}[Feasibility of Shewhart]\label{lemma:feasible}
Shewhart test achieves the false alarm constraints of~\eqref{eq:p1} and~\eqref{eq:p2} with equality.
\end{lemma}
\begin{proof}
For the Shewhart test we have
\begin{align}
\label{feas1}
\mathbb{E}_\infty\{\tau_{\rm s}\} &= \sum_{t=1}^\infty \mathbb{P}_\infty(\tau_s\geq t)\\
\label{feas2}
& = \sum_{t=1}^\infty [1-\mathbb{P}_\infty(\ell_1\geq\alpha)]^{t-1}\\
\label{feas3}
&= \frac{1}{\mathbb{P}_\infty(\ell_1\geq\alpha)}\\
\label{feas4}
&= \eta\ ,
\end{align}
where~\eqref{feas2} is due to the fact that at each time we stop we probability $\mathbb{P}_\infty(\ell_1<\alpha)=1-\mathbb{P}_\infty(\ell_1\geq\alpha)$,~\eqref{feas3} is the result of the infinite sum, and~\eqref{feas4} holds because of~\eqref{eq:threshold}.
\end{proof}
\begin{lemma}\label{lemma:upper_bound}
Shewhart test is the solution to 
\begin{align}\label{eq:p3}
\sup_{\tau\;:\;\bbe_\infty\{\tau\}=\eta} \frac{\bbe_\infty\{\ell_\tau\}}{\bbe_{\infty}\{\tau\}}\ .
\end{align}
\end{lemma}
\begin{proof}
The Lagrangian  corresponding to the constrained problem in~\eqref{eq:p3} is
\begin{align}
\mathcal{L}(\tau) &\dff \bbe_\infty\{\ell_\tau-\lambda\tau\}\ .
%&= (1-\pi_0)\cdot\bbe_\infty\{\ell_\tau-\lambda\tau\;|\;\tau>0\}\ .
\end{align}
%Optimizing this Lagrangian follows the standard stopping rule techniques~\cite{ShiryaevOptimal}, and we show that Shewhart optimizes this objective. 
We show that the Shewhart test is the solution to this unconstrained problem. To this end, leveraging the standard stopping rule techniques~\cite{ShiryaevOptimal} we define 
\begin{align}
G_t(\F_t)\dff\max\big\{\ell_t\, ,\, -\lambda+\bbe_\infty\{G_{t+1}(\F_{t+1})\,|\,\F_t\}\big\}
\end{align}
as the maximal utility function at each time $t$, where $\ell_t$ is the utility if we stop at time $t$, and $-\lambda+\bbe_\infty\{G_{t+1}(\F_{t+1})|\F_t\}$ is the return of taking one more measurement from the sequence. The sampling process stops as soon as the  
\begin{align}
\ell_t \geq -\lambda+\bbe_\infty\{G_{t+1}(\F_{t+1})\,|\,\F_t\}\ .
\end{align}
It can be readily verified through backward induction that the maximal utility function depends on $\F_t$ only through $\ell_t$. Furthermore, backward induction can be used to show that
\begin{align}
G_t(\ell_t)&\dff\max\big\{\ell_t\, ,\, -\lambda+\bbe_\infty\{\ell_{t+1}\}\big\}\\
&=\max\big\{\ell_t\, ,\, C\big\}\ .
\end{align}
where $C$ is a constant. Hence, the optimal solution reduces to comparing the likelihood ratio of the current measurement with a constant, which is the Shewhart test.
\end{proof}
\begin{theorem}\label{lemma:Sh}
The Shewhart test with the stopping time and threshold given in~\eqref{eq:stop} and~\eqref{eq:threshold}, respectively, optimizes~\eqref{eq:p3} and, therefore, is the optimal solution to~\eqref{eq:p1} and~\eqref{eq:p2}, i.e., 
%Under the Shewhart test, for the objective function $\L_{\rm P}(\tau_s)$ we have
\begin{align}
\L_{\rm L}(\tau_s)= \sup_{\tau\;:\;\bbe_\infty\{\tau\}\geq\eta} \L_{\rm L}(\tau)\ .
\end{align}
\end{theorem}
\begin{proof}
%For the Shewhart test we have
%\begin{align}
%%\mathbb{P}_\gamma(\tau_s=\gamma_i\;|\;\mathcal{F}_{\gamma_i-1}\;,\;\tau\geq \gamma_i)=\mathbb{P}(\ell_1\geq\alpha\;|\;\{1\}\in\gamma)\nonumber
%\L_{\rm P}(\tau_s)=\mathbb{P}_\gamma(\ell_{\gamma_i}\geq\alpha\;,\;\text{for some }\gamma_i\in\gamma)\ ,
%\end{align}
%which shows that Shewhart is an equalizer. 
From Lemma~\ref{lemma:upper_bound} for any feasible stopping time $\tau$ that meets the false alarm constraint with equality we have
\begin{align}\label{eq:sha}
\mathcal{L}(\tau)=\bbe_\infty\{\ell_{\tau}-\lambda\tau\}\leq\bbe_\infty\{\ell_{\tau_s}-\lambda\tau_s\}= \mathcal{L}(\tau_s)\ ,
\end{align}
and therefore, $\bbe_\infty\{\ell_{\tau}\}\leq\bbe_\infty\{\ell_{\tau_s}\}$. Now, we have
\begin{align}
\label{eq:up1}
\sup_{\tau:\bbe_\infty\{\tau\}\geq\eta}\L_{\rm L}(\tau) &\leq \eta^{-1}\sup_{\tau:\bbe_\infty\{\tau\}=\eta}\bbe_\infty\{\ell_{\tau}\}\\
\label{eq:up2}
&\leq \eta^{-1}\sup_{\tau:\bbe_\infty\{\tau\}=\eta}\bbe_\infty\{\ell_{\tau_s}\}\\
%&= \mathbb{P}_\gamma(\ell_{\gamma_i}\geq\alpha\;,\;\text{for some }\gamma_i\in\gamma)\\
\label{eq:up3}
&= \L_{\rm L}(\tau_s)\ ,
\end{align}
where~\eqref{eq:up1} results from replacing $\L_{\rm L}(\tau)$ with its upper bound,~\eqref{eq:up2} holds due to~\eqref{eq:sha}, and~\eqref{eq:up3} is due to Lemma~\ref{lemma:upper_bound}. Since the upper bound on the objective is achieved, the proof is concluded.
%by setting $\lambda=\mathbb{P}_0(\ell_1\geq\alpha)-\alpha\mathbb{P}_\infty(\ell_1\geq\alpha)$, we have
%\begin{align}
%\mathcal{L}(\tau_s)=\bbe_\infty\{\ell_{\tau_s}-\lambda\tau_s\}=\alpha\ ,
%\end{align}
%where we have used the fact that 
%\begin{align}
%\bbe_\infty\{\ell_{\tau_s}\}=\frac{\mathbb{P}_0(\ell_1\geq\alpha)}{\mathbb{P}_\infty(\ell_1\geq\alpha)}=\eta\mathbb{P}_0(\ell_1\geq\alpha)\ .
%\end{align}
%This indicates that the Shewhart test achieves the upper bound in~\eqref{eq:up}, and, consequently, is optimal.
\end{proof}
%Based on the results from the lemmas above, the optimality of Shewhart test is established in the following theorem.
%\begin{theorem}
%The Shewhart test with the stopping time and threshold given in~\eqref{eq:stop} and~\eqref{eq:threshold}, respectively, optimizes~\eqref{eq:p3} and, therefore, is the optimal solution to~\eqref{eq:p1} and~\eqref{eq:p2}.
%\end{theorem}
%\begin{proof}
%
%\end{proof}
The popularity of Shewhart test is mostly due to its simple implementation. At each time $t$ we take a new measurement from the sequence, form its likelihood ratio, and compare the likelihood ratio with a fixed pre-specified upper threshold. We stop the process and declare a change the first time the likelihood ratio exceeds the threshold. 
%Figure~\ref{fig:Shewhart} illustrates the evolution of likelihood ratio and the stopping time.

%\vspace{-.1 in}
\section{Numerical Results}
\label{sec:sim}

\begin{figure}
\centering
\includegraphics[width=2.8in]{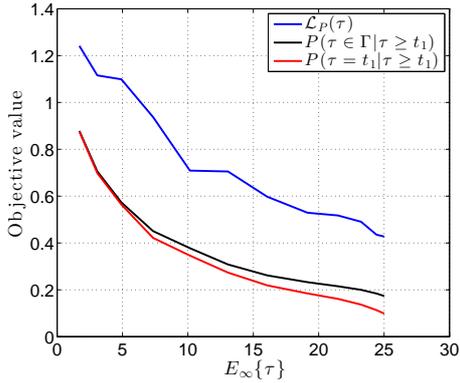}
%\vspace{-0.15in}
\caption{The probability of detecting the first or any change-point.}
\label{fig:obj}
%\vspace{-0.17in}
\end{figure}

\begin{figure}
\centering
\includegraphics[width=2.8in]{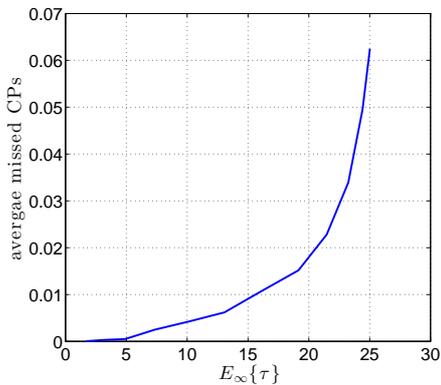}
%\vspace{-0.15in}
\caption{The average number of missed change-points before detection.}
\label{fig:miss}
%\vspace{-0.2in}
\end{figure}

\begin{figure*}
\centering
\includegraphics[width=0.99\textwidth]{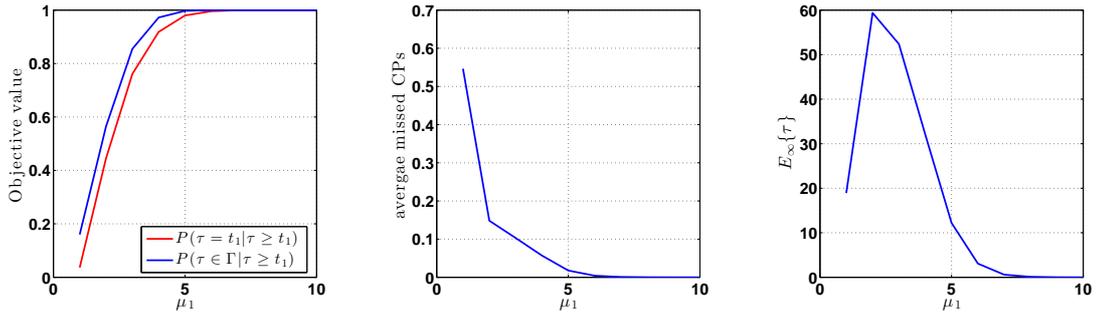}
%\vspace{-0.15in}
\caption{Performance of the Shewhart test for the modified Pollak criterion.}
\label{fig:KL}
%\vspace{-0.2in}
\end{figure*}

In this section we numerically evaluate the performance of the Shewhart test corresponding to the Pollak-like criterion. Specifically, we are interested in assessing the impact that our choice of the selected agility criterion has on the test performance. Based on our Pollak-like criterion, we can afford to miss some of the change-points in favor of making more confident decisions about the onset of a change. In order to quantify how this specific objective function affects the test performance, for the Shewhart test we compare the probability of detecting the first change-point with the probability of detecting any change-point. We, also, provide the average number of missed change-points by the Shewhart test.

To this end, we consider a sequence of $10^5$ random variables, where the nominal and alternative distributions are unit-variance Gaussian distributions with mean values $0$ and $1$, respectively. There exist $1000$ change-points in the sequence, each with the duration $T=1$. Figure~\ref{fig:obj} compares the conditional probability of detecting the first change-point with that of detecting any change-point. It is observed that when we have a more stringent constraint on the false alarm rates, i.e., the average run length to a false alarm increases, the detection probability decreases since we want to raise fewer false alarms. Also, the ratio gap between these two objective function becomes more significant. This is due to the fact that in our objective function, we can afford to wait for a more reliable decision about the occurrence of a change-point. Figure~\ref{fig:miss} illustrates the average number of missed change-points in our setting. It is observed that for larger average run length to a false alarm we miss more change-points in order to detect one of them more reliably.

In order to evaluate the effect of similarity level of the pre-change and post-change distributions, in Fig.~\ref{fig:KL} we repeat the simulation for various values of the mean for the post-change distribution. It is observed that by increasing the mean, which is equivalent to more less similarity to the pre-change distribution, the average number of missed change-points decreases and the Shewhart test detects the first change-point more reliably.

\section{Conclusion}
\label{sec:conclusion}

We have analyzed the problem of quickest search for change-points when the changes are not persistent. We have considered a setting in which a sequence of random variables might undergo  multiple change-points and after each change-point it returns to the nominal distribution. Both the pre-change and post-change distributions are known and the objective is to identify one of these change-points in real-time, i.e., by observing the first measurement generated according to the post-change distribution, while controlling the false alarm rate in parallel. To this end, we have considered a probability maximizing approach in a minimax setting. We have shown that the Shewhart test, which is a likelihood ratio test based on the current observed measurement, is exactly optimal.

\bibliographystyle{IEEEtran}
\bibliography{ISIT18-1}

\end{document}